\documentclass{article}
\usepackage{xcolor,amsmath,amssymb,tikz,graphics,url}

\newcommand{\R}{\mathbb{R}}
\newcommand{\Z}{\mathbb{Z}}
\newcommand{\N}{\mathbb{N}}

\DeclareMathOperator{\Card}{Card}
\def\u(#1){\underline{#1\!}\,}
\def\1{\mathbf{1}}
\newcommand{\edge}[1]{\stackrel{#1}{\rightarrow}}
\newcommand{\longedge}[1]{\stackrel{#1}{\longrightarrow}}
\newcommand{\A}{\mathcal{A}}
\newcommand{\B}{\mathcal{B}}
\newcommand{\Tr}{\mathcal{T}}
\newcommand{\GD}{\mathcal D}
\newcommand{\GH}{\mathcal H}
\newcommand{\GR}{\mathcal R}
\newcommand{\GL}{\mathcal L}
\newcommand{\Langle}{\langle\!\langle}
\newcommand{\Rangle}{\rangle\!\rangle}

\usepackage{theorem}
\newtheorem{theorem}{Theorem}

\newtheorem{proposition}[theorem]{Proposition}

{\theorembodyfont{\rmfamily}%
  \newtheorem{example}[theorem]{Example}
   }
\newenvironment{proof}{\noindent\textit{Proof.}}
{\QED\vskip\theorempostskipamount} 

\def\petitcarre{\vrule height4pt width 4pt depth0pt}
\def\QED{\relax\ifmmode\eqno{\hbox{\petitcarre}}\else{%
  \unskip\nobreak\hfil\penalty50\hskip2em\hbox{}\nobreak\hfil
  \petitcarre
  \parfillskip=0pt \finalhyphendemerits=0\par\smallskip}
  \fi}

\title{The Degree of a Finite Set of Words}

\author{Dominique Perrin,
Andrew Ryzhikov}

\begin{document}

\maketitle

\begin{abstract}
We generalize the notions of the degree and composition from uniquely decipherable codes to arbitrary finite sets of words.
We prove that if $X=Y\circ Z$ is a composition of finite sets of words with $Y$ complete, then $d(X)\le d(Y) \cdot d(Z)$, where $d(T)$ is the degree of $T$. We also show that a finite set is synchronizing if and only if its degree equals one. 

This is done by considering, for an arbitrary finite set $X$ of words, the transition monoid of an automaton recognizing $X^*$ with multiplicities. We prove a number of results for such monoids, which generalize corresponding results for unambiguous monoids of relations.
\end{abstract}

\section{Introduction}

Let $X$ be a set of finite words. The set $X^*$ of all concatenations of words in $X$ (often called the Kleene star of $X$) plays an important role in formal languages theory and its applications. The set $X$ often represents a dictionary or a code transmitted over a channel, so the case where $X$ is finite is especially important. In general, a word in $X^*$ can have several different factorizations over $X$, and it is useful to understand the relations between them. A word $w$ is called \emph{synchronizing} for $X$ if for any words $u, v$ such that $uwv \in X^*$ we have $uw, wv \in X^*$. In particular, we get that any word in $X^*$ containing $ww$ as a factor, that is, any word of the form $uwwv$, has a factorization where $uw$ and $wv$ are both in $X^*$, and thus can be factorized separately. 
A set which admits a synchronizing word is also called \emph{synchronizing}. A set $X$ is called \emph{complete} if every word over the same alphabet occurs as a factor of a word in $X^*$.

Synchronizing words are studied a lot for uniquely decipherable codes (see e.g., Chapter 10 of \cite{BerstelPerrinReutenauer2009}). A set $X$ of words is called a \emph{uniquely decipherable code} (often also called a \emph{variable length code}) if every word has at most one factorization over $X$. Such codes play a crucial role in the theory of data compression and transmission \cite{BerstelPerrinReutenauer2009}.

Provided a set $Z$ of words such that $X \subset Z^*$, one can rewrite $X$ using $Z$ as the alphabet, thus resulting in a new set~$Y$. The representation $X = Y \circ Z$ is then called a decomposition of $X$, and the converse process of obtaining $X$ is called composition. Decomposition of a set allows to represent it by using simpler sets as building blocks, while preserving many properties of the initial one. Conversely, compositions of codes allow to construct more complicated codes by using simple ones, so they are interesting on their own. In particular, the composition of two uniquely decipherable codes is again a uniquely decipherable code~\cite{BerstelPerrinReutenauer2009}. For any injective morphism $\alpha: A^* \to B^*$, $\alpha(A)$ is a code, and each code can be obtained as the image of $A$ for some $A$ and $\alpha$ \cite{BerstelPerrinReutenauer2009}. Compositions of codes are then nothing more than compositions of injective morphisms between free monoids. The notion of composition of two arbitrary finite sets of words is also natural as it corresponds to the composition of arbitrary morphisms.

\subparagraph*{Our contributions.} In this paper, we transfer the notions of composition, degree and group from uniquely decipherable codes to arbitrary finite sets of words. This extends the presentation of~\cite{BerstelPerrinReutenauer2009} made for uniquely decipherable codes.

Provided a finite set $X$ of words, we associate a special automaton $\mathcal{A}$ (called the flower automaton) recognizing $X^*$ with multiplicities. Let $S$ be the set of fixed points of an idempotent $e$ of minimum rank in the transition relation of $\mathcal{A}$, and $\Gamma$ be the set of strongly connected components of $S$. We consider a permutation group $G_e$ acting on $\Gamma$. We show that all such groups are equivalent for idempotents of minimum rank (Theorem \ref{theorem9.3.10}). Moreover, we show that for a given $X$ all these groups are equivalent for any trim automaton recognizing $X^*$ with multiplicities (Proposition \ref{proposition9.5.1}). Thus this group is an invariant of a set. We introduce the degree $d(X)$ of $X$, which is the minimum rank of elements in the transition monoid of $\mathcal{A}$. We then show that synchronizing sets are exactly sets of degree one (Proposition \ref{prop-degree-one}). As our main contribution, we use the obtained results to show that for a composition $X = Y \circ Z$ of two finite sets $Y, Z$ such that $Y$ is complete we have $d(X) \le d(Y) \cdot d(Z)$ (Theorem \ref{theoremDegrees}). 

For a finite set $X$, all these results were previously known only for the special case of $X$ being a uniquely decipherable code
with the equality $d(X)=d(Y)d(Z)$
instead of an inequality \cite{BerstelPerrinReutenauer2009}. Our generalization to the case of an arbitrary finite set requires more complicated proofs. In particular, for uniquely decipherable codes it is enough to consider a trim unambiguous automaton recognizing $X^*$ (which is a cornerstone of the theory), while in our case we need a trim automaton recognizing $X^*$ with multiplicities. Intuitively, such automata count the number of factorizations over $X$, and thus they are unambiguous when $X$ is a uniquely decipherable code.
The technical difficulties then begin with the replacement of unambiguous
monoids of relations by arbitrary monoids of relations.
Indeed, the multiplication of matrices there is different
from the result over the Boolean semiring. In particular,
the representation
of maximal subgroups by permutations is still possible but more complicated.

\subparagraph*{Motivation and related results.} Larger classes of codes are considered both in theory and in practice. Particular examples include multiset and set decipherable codes. A set $X$ of words is called a \emph{multiset} \cite{Guzman1999} (respectively, \emph{set} \cite{Lempel1986}) \emph{decipherable code} if every factorization of a word into codewords provides the same multiset (respectively, set) of codewords. Such codes are used if one needs to transmit only the frequencies (or the fact of occurrences) of elements, but the order of these elements does not matter. Lempel \cite{Lempel1986} reports online compilations of inventories, construction of histograms, or updating of relative frequencies as particular examples. An important property of multiset decipherable codes is that there exist examples of such codes with Kraft-McMillan sum more than one, which shows that such codes can be more efficient than uniquely decipherable codes \cite{Restivo1989}. An even wider class is that of numerically decipherable codes, which are sets with the property that every factorization of a word over such set has the same number of codewords \cite{Weber1996}. A similar setting of multivalued encodings allows to have several different codewords for the same symbol \cite{Capocelli1994}. In view of that, the transit of results from uniquely decipherable codes to arbitrary sets is interesting.

Another motivation for studying factorizations of words in $X^*$ for an arbitrary finite set~$X$ is the area of static dictionary compression, where one looks for some specific factorization of a text over some finite dictionary \cite{Bell1990}. The dictionary does not have to be a uniquely decipherable code, thus a text can have several different factorizations. In this case, it is useful to know the relation between different factorizations. The parallel version of this problem is also considered \cite{Nagumo1995}. In \cite{Clement2005} a fast algorithm for checking if a given word $w$ belongs to $X^*$ is suggested. If the answer is positive, it also provides a factorization of $w$ over $X$.

Only few results are known about decompositions and synchronization of arbitrary sets of words. The defect theorem states that every finite set of words which is not a uniquely decipherable code can be decomposed over a set of smaller size \cite{Berstel1979Defaut}. A survey of different generalizations of this theorem is presented in \cite{Harju2004}. Synchronization in arbitrary monoids was studied in \cite{Carpi2017} and \cite{Deluca1979}. Other properties of factorizations are studied in \cite{Rest1983, Schutzenberger1979}. 

\subparagraph*{Organization of the paper.} 
To transfer the results from uniquely decipherable code to arbitrary finite sets of words, we first set a correspondence with an adequate
class of automata, namely automata recognizing with multiplicities (Sections \ref{sect-automata} and \ref{sect-tranducers}).
Then we introduce the notion of a composition for arbitrary
finite sets of words (Section \ref{sect-composition}).
We extend the theory of unambiguous
monoids of relations by the theory of arbitrary monoids of relations (Section \ref{sect-relations}), and generalize the notion of the group $G(X)$ and the degree $d(X)$
of a finite set $X$
of words (Section \ref{sect-group}). In this way, as for codes,
a set is synchronizing if and only if it is of degree $1$ (Section \ref{sect-synchronization}).
As the main result, we prove that if $X=Y\circ Z$ with $Y$ complete, then
$d(X)\le d(Y) \cdot d(Z)$
(Section \ref{sect-groups-composition}). In Section \ref{sect-codes} we show that if we require $Y$ to be complete, we do not get any new decompositions of a uniquely decipherable code other than into two uniquely decipherable codes.

\paragraph{Acknowledgements}
A preliminary version of this paper
appeared in~\cite{PerrinRyzhikov2020}.
We thank Jean-Eric Pin and Jacques Sakarovitch for references concerning the composition of automata and transducers.

\section{Automata} \label{sect-automata}
We denote by $A^*$ the free monoid on a finite alphabet $A$,
by $1$ the empty word, and by $A^+$ the set $A^* \setminus \{1\}$. For notions not defined in this section see \cite{BerstelPerrinReutenauer2009}.

Let $\A=(Q,i,t)$ be an automaton on the alphabet $A$ with $Q$ as set
of states, $i$ as initial state and $t$ as terminal state
(we will not need to have several initial or terminal states).
We do not specify in the notation the set of edges, which are
triples $(p,a,q)$ with two states $p,q\in Q$ and a label
$a\in A$ denoted $p\edge{a}q$. We form paths as usual
by concatenating consecutive edges. An automaton is called \emph{trim} if there exists a path from $i$ to every state, and from every state to $t$.

The \emph{language recognized} by $\A$, denoted $L(\A)$,
is the set of words in $A^*$
which are labels 
of paths from $i$ to $t$. There can be several paths from $i$ to $t$
for a given label, and this motivates the introduction of multiplicities.

For a semiring $K$, a $K$-subset of $A^*$ is
a map from $A^*$ into $K$. The value of a $K$-subset~$X$ at $w$ is called
its \emph{multiplicity} and denoted $(X, w)$.
We denote by $K\Langle A\Rangle$ the semiring
of $K$-subsets of $A^*$ and by $K\langle A\rangle$
the set of corresponding polynomials, that is
the $K$-subsets with a finite number of words
with nonzero multiplicity (on these notions, see~\cite{Eilenberg1974}).

If $X,Y$ are $K$-subsets, then $X+Y$ and $XY$ are the $K$-subsets defined by 
\begin{displaymath}
(X+Y,w)=(X,w)+(Y,w), \quad (XY,w)=\sum_{w=uv}(X,u)(Y,v).
\end{displaymath}
Moreover, if $X$ does not have a constant term, that is, if $(X,1)=0$, then $X^*$ is the $K$-subset
\begin{displaymath}
X^*=1+X+X^2+\ldots
\end{displaymath}
Since $X$ has no constant term, for every word $w$, the number of nonzero
terms $(X^n,w)$ in the sum above is finite and thus $X^*$ is well-defined.

For a set $X\subset A^*$, we denote by $\u(X)$ the characteristic series
of $X$, considered as an $\N$-subset. It is easy to
verify that for $X\subset A^+$, the mutiplicity
of $w\in A^*$ in $\u(X)^*$ is the number of factorizations
of $w$ in words of $X$.

For an automaton $\A=(Q,i,t)$ on the alphabet $A$,
we denote by $|\A|$ its behaviour,
which is an element of $\N\Langle A\Rangle$.
It is the $\N$-subset of $A^*$ such that the multiplicity
of $w\in A^*$ in $|\A|$ is the number of paths from $i$ to $t$
labeled $w$ in $\A$.

We denote by $\mu_\A$ the morphism from $A^*$ into the
monoid of $Q\times Q$-matrices with integer coefficients
defined for $\mu_\A(w)_{p,q}$ as the number of paths from $p$ to $q$ labeled by $w$.
Thus, the multiplicity of $w$ in $|\A|$ is $(|\A|,w)=\mu_\A(w)_{i,t}$.

Given a set $X\subset A^+$, we say that the automaton $\A$ \emph{recognizes}
$X^*$ \emph{with multiplicities} if the behaviour
of $\A$ is the multiset assigning to $x$ its
number of distinct factorizations in $X$. Formally, $\A$
recognizes $X^*$ with multiplicities if $|\A|=\u(X)^*$.

\begin{example}\label{example{a,a^2}}
Let $X=\{a,a^2\}$. The number of factorizations of $a^n$ in words of $X$
is the Fibonacci number $F_{n+1}$ defined by $F_0=0$, $F_1=1$
and $F_{n+1}=F_n+F_{n-1}$ for $n\ge 1$.
The automaton $\A$ represented in Figure~\ref{figureFibo}
recognizes $X^*$ with multiplicities, that is $|\A|=(a+a^2)^*$.
\begin{figure}[hbt]
  \centering

  \includegraphics{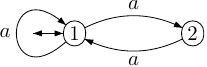}
\caption{An automaton recognizing $X^*$ with multiplicities.}
\label{figureFibo}
\end{figure}

We have indeed for every $n\ge 1$,
\begin{displaymath}
\mu_\A(a^n)=\begin{bmatrix}F_{n+1}&F_n\\F_n&F_{n-1}\end{bmatrix}
\end{displaymath}
\end{example}
For an automaton $\A=(Q,i,t)$ on the alphabet $A$,
we denote by $\varphi_\A$ the morphism from $A^*$ onto
the monoid of transitions of $\A$. Thus, for $w\in A^*$,
$\varphi_\A(w)$ is the Boolean $Q\times Q$-matrix defined by
\begin{displaymath}
\varphi_\A(w)_{p,q}=\begin{cases}1&\mbox{ if $p\edge{w}q$},\\0&\mbox{otherwise}
\end{cases}
\end{displaymath}
Let $X\subset A^+$ be a finite set of words on the alphabet $A$.
The \emph{flower automaton} of $X$ is the following automaton.
Its 
set of states is the subset $Q$ of $A^*\times A^*$ defined as
\begin{displaymath}
Q=\{(u,v)\in A^+\times A^+\mid uv\in X\}\cup (1,1).
\end{displaymath}
We often denote $\omega=(1,1)$.
There are four type of edges labeled by $a\in A$
\begin{eqnarray*}
(u,av)&\longedge{a}&(ua,v)\quad\text{for $uav\in X$, $u,v\ne1$}\\
\omega&\longedge{a}&(a,v)\quad\text{for $av\in X$, $v\ne 1$}\\
(u,a)&\longedge{a}&\omega\quad\text{for $ua\in X$, $u\ne 1$}\\
\omega&\longedge{a}&\omega\quad\text{for $a\in X$.}
\end{eqnarray*}
The state $\omega$ is both initial and terminal.

The proof of the following result is straightforward.
It generalizes the fact that the flower automaton
of a code recognizes $X^*$ and is unambiguous (see Theorem 4.2.2 in \cite{BerstelPerrinReutenauer2009}).

\begin{proposition}
For any finite set $X\subset A^+$,
the flower automaton of $X$ recognizes $X^*$ with multiplicities.
\end{proposition}

\begin{example}\label{example{a,ab,ba}}
Let $X=\{a,ab,ba\}$. The flower automaton of $X^*$ is represented in Figure~\ref{figureFlower1}. 
As an example, there are two paths from $\omega$
to $\omega$ labeled $aba$, corresponding to the two factorizations $(a)(ba)=(ab)(a)$.
\begin{figure}[hbt]
\centering
\includegraphics{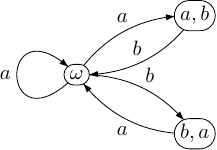}
\caption{The flower automaton of $X$ (Example~\ref{example{a,ab,ba}}).}\label{figureFlower1}
\end{figure}
\end{example}

A more compact version of the flower automaton is the \emph{prefix automaton}
$\A=(P,1,1)$ of a finite set $X\subset A^+$. Its set of states is the set $P$ of proper
prefixes of $X$ and its edges are the $p\edge{a}pa$ for every
$p\in P$ and $a\in A$ such that $pa \in P$ and the $p\edge{a}1$
such that $pa\in X$. It also recognizes $X^*$ with multiplicities.

\begin{example}\label{example-aa-aaa}
Let $X=\{a^2,a^3\}$. The flower automaton of $X$ is shown in Figure~\ref{figureFlower1bis} on the left and its prefix automaton on the right.
\begin{figure}[hbt]
  \centering
  \includegraphics{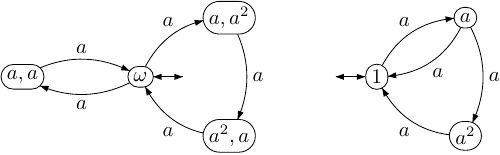}
\caption{The flower automaton and the prefix automaton of $X$ (Example~\ref{example-aa-aaa}).}
\label{figureFlower1bis}
\end{figure}
\end{example}

A \emph{reduction} from an automaton $\A=(P,i,t)$ onto an automaton
$\B=(Q,j,u)$ is a surjective map $\rho:P\to Q$ such that
$\rho(i)=j$, $\rho(t)=u$ and such that 
for every 
$q,q'\in Q$ and $w\in A^*$, there is a path $q\edge{w}q'$
in $\B$  if and only if there is a path $p\edge{w}p'$ in $\A$ for some $p,p'\in P$ with $\rho(p)=q$ and $\rho(p')=q'$

The reduction is \emph{sharp} if $\rho^{-1}(j)=\{i\}$ and $\rho^{-1}(u)=\{t\}$.

\begin{proposition}
Let $\rho$ be a reduction from $\A=(P,i,t)$ onto $\B=(Q,j,u)$.
Then $L(\A)\subset L(\B)$, with equality if $\rho$ is sharp.
\end{proposition}

The term reduction is the one used in~\cite{BerstelPerrinReutenauer2009}
and it is not standard but captures the general idea of a covering.
The term conformal morphism is the one used in \cite{Sakarovitch2009}.
The following statement replaces \cite[Proposition 4.2.5]{BerstelPerrinReutenauer2009}.

\begin{proposition}\label{proposition4.2.5}
Let $X\subset A^+$ be a finite set which is the minimal
generating set of $X^*$. For each trim automaton
$\B=(Q,i,i)$ recognizing $X^*$ with multiplicities, there is a sharp reduction
from the
flower automaton of $X$ onto $\B$.
\end{proposition}
\begin{proof}
Let $\A=(P,\omega,\omega)$ be the flower automaton of $X$.
We define a map $\rho:P\to Q$ as follows. We set first $\rho(\omega)=i$.
Next, if $(u,v)\in P$ with $(u,v)\ne\omega$, then $uv\in X$. Since $X$ is the minimal
generating set of $X^*$,  there is only one factorization of $uv$
into elements of $X$.  Since $\B$
recognizes $X$ with multiplicities, there is only one path 
$i\edge{u}q\edge{v}i$
in $\B$. We define $\rho\bigl((u,v)\bigr)=q$.

It is straightforward to verify that $\rho$ is a reduction. Assume first
that $q\edge{w}q'$ in $\B$. Let $i\edge{u}q$ and $q'\edge{v'}i$ be simple paths,
that is not passing by $i$ except at the origin or the end. Then 
$i\edge{uwv'}i$ and thus $uwv'=x_1x_2\cdots x_n$ with $x_i\in X$,
$u$ a proper prefix
of $x_1=uv$ and $v'$ a proper suffix of $x_n=u'v'$. Thus $\rho\bigl((u,v)\bigr)=q$
and $\rho\bigl((u',v')\bigr)=q'$. Since $w=vx_2\cdots x_{n-1}u'$,
we have in $\A$ a path
$(u,v)\edge{w}(u',v')$. Conversely, consider a path
$(u,v)\edge{w}(u',v')$ in $\A$. If the path does not pass
by $\omega$, then $u'=uw$, $v=wv'$
and we have a path $q\edge{w}q'$
in $\B$ with $\rho\bigl((u,v)\bigr)=q$ and $\rho\bigl((u',v')\bigr)=q'$.
Otherwise, the path
decomposes in $(u,v)\edge{v}\omega\edge{x}\omega\edge{v'}(u',v')$ with $x\in X^*$. Since $\B$ recognizes $X^*$, we have
a path $i\edge{x}i$ in $\B$ and thus also a path $q\edge{w}q'$
with $q=\rho\bigl((u,v)\bigr)$ and $q'=\rho\bigl((u',v')\bigr)$.
\end{proof}

The statement above is false if $X$ is not the minimal generating
set of $X^*$, as shown by the following example.

\begin{example}
Let $X=\{a,a^2\}$. There is no sharp reduction from the 
automaton of Figure~\ref{figureFibo} to the one-state
automaton recognizing $X^*=\{a\}^*$.
\end{example}

The statement is also false if the automaton $\B$ recognizes $X^*$, but does not recognize $X^*$ with multiplicities, as shown by the following example.

\begin{example}
Let $X=\{a^2\}$. The flower automaton of $X$ is represented in Figure~\ref{figureNoRed} on the left. There is no reduction to the automaton
represented on the right which also recognizes $X^*$
(but not with multiplicities).
\begin{figure}[hbt]
\centering
\includegraphics{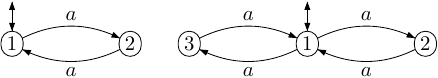}
\caption{Two automata recognizing $X^*$.}
\label{figureNoRed}
\end{figure}
\end{example}

\section{Transducers} \label{sect-tranducers}
A literal \emph{transducer} $\Tr=(Q,i,t)$
on a set of states $Q$ with an \emph{input alphabet} $A$
and an \emph{output alphabet} $B$ is defined by a set
of edges $E$ which are of the form $p\longedge{(a,v)}q$ with $p,q\in Q$,
$a\in A$ and $v\in B\cup \{1\}$. The \emph{input automaton} associated
with a transducer is the automaton with the same set
of states and edges but with the output labels removed.

The relation \emph{realized} by the transducer $\Tr$ is the set of pairs $(u,v)\in A^*\times B^*$ such that there is a path from $i$ to $t$ labeled $(u,v)$.
We denote by $\varphi_\Tr$ the morphism from $A^*$ to the monoid
of $Q\times Q$-matrices with elements in $\N\langle B\rangle$
defined for $u\in A^*$ and $p,q\in Q$ by $\textstyle \varphi_\Tr(u)_{p,q}=\sum_{p\longedge{u|v}q  }v$.

Let $X\subset A^+$ be a finite set. Let $\beta:B^*\to A^*$ be a \emph{coding morphism} for $X$, that is, a morphism whose restriction to $B$ is a bijection  onto $X$.
The \emph{decoding relation} for $X$ is the relation $\gamma=\{(u,v)\in A^*\times B^*\mid u=\beta(v)\}$. 
A \emph{decoder} for $X$ is a literal transducer which realizes
the decoding relation.
The \emph{flower transducer} associated to $\beta$ is the literal
tranducer built on the flower automaton of $X$ by adding an output
label $1$ to each edge $\omega\edge{a}(a,v)$
or $(u,av)\edge{a}(ua,v)$ and an output label $b$ to each edge
$\omega\edge{a}\omega$ such that $a\in X$ with $\beta(b)=a$ or
$(u,a)\edge{a}\omega$ such that $ua=x\in X$ 
with $\beta(b)=x$.

\begin{proposition}
For every finite set $X \subset A^+$ 
with a coding morphism $\beta$, the flower transducer
associated to $\beta$ is a decoder
for $X$. 
\end{proposition}

\begin{example}
Let $X=\{a,ab,ba\}$ and let $\beta:u\to a, v\to ab, w\to ba$.
The flower transducer associated to $\beta$ is
represented in Figure~\ref{figureDecoder}.
\begin{figure}[hbt]
\centering
\includegraphics{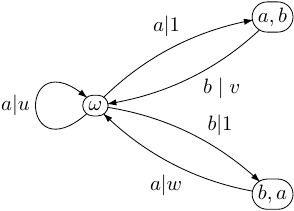}
\caption{The flower transducer associated to $\beta$.}\label{figureDecoder}
\end{figure}
One has
\begin{displaymath}
\varphi_\Tr(a)=\begin{bmatrix}u&1&0\\0&0&0\\w&0&0\end{bmatrix}
\text{ and }
\varphi_\Tr(b)=\begin{bmatrix}0&0&1\\v&0&0\\0&0&0\end{bmatrix}.
\end{displaymath}
\end{example}

The \emph{prefix transducer} $\Tr=(P,1,1)$ is the 
same modification of the prefix automaton.
Its states are the proper prefixes of the elements of $X$.
There is an edge $p\edge{a|1}pa$ for every prefix $p$
and every letter $a$ such that $pa \in P$, and an edge $p\edge{a|b}1$ for every
prefix $p$ and letter $a$ such that $pa=\beta(b)\in X$.
Thus the input automaton of the prefix transducer of $X$
is the prefix automaton of $X$.

Let $\B=(Q,j,j)$ be an automaton on the alphabet $B$ 
and let $\Tr=(P,i,i)$ be a literal transducer with the input alphabet $A$
and the output alphabet $B$.
We build an automaton $\A=\B\circ \Tr$ on the alphabet $A$ as follows. Its
set of states is $Q\times P$ and for every $a\in A$, the matrix
$\varphi_\A(a)$ is obtained by replacing in $\varphi_\Tr(a)$
the word $w=\varphi_\Tr(a)_{p,q}$ by the matrix $\varphi_\B(w)$.
The initial and terminal state is $(j,i)$.
The automaton $\A$ is also called the \emph{wreath product}
of $\B$ and $\Tr$ (see~\cite{Eilenberg1976}). The word $1$ is replaced by the identity matrix, and $0$ is replaced by the zero matrix of appropriate size. An example of $\A=\B\circ \Tr$ is provided in Example \ref{example-composition}.

\section{Composition} \label{sect-composition}
Let $Y\subset B^+$ and $Z\subset A^+$ be finite sets of words such that
there exists a bijection $\beta:B\to Z$.
Two such sets
are called \emph{composable}. Then $X=\beta(Y)$ is called
the \emph{composition} of $Y$ and $Z$ through $\beta$, where $\beta(Y) = \{\beta(y) \mid y \in Y\}$ with $\beta$ naturally extended to the mapping $B^* \to Z^*$. We denote
$X=Y\circ_\beta Z$. We also denote $X=Y\circ Z$ when $\beta$ is
clear. We say that $X=Y\circ Z$ is a \emph{decomposition}
of~$X$.

\begin{example}\label{exampleComposition}
Let $Y=\{u,uw,vu\}$ and $Z=\{a,ab,ba\}$ with $\beta:u\to a, v\to ab, w\to ba$.
Then $X=Y\circ_\beta Z =\{a,aba\}$.
\end{example}

A decomposition $X=Y\circ_\beta Z$ of a finite set $X$ is \emph{trim} if every letter of $B$ appears in a word of $Y$ and every word in $X$ is obtained in a unique way from words in $Y$, that is, if the restriction of $\beta$ to $Y$ is injective. For any decomposition
$X=Y\circ Z$, there are $Y'\subset Y$ and $Z'\subset Z$ such that $X=Y'\circ Z'$ is trim.
Indeed, if $x\in X$ has two decompositions in words of $Z$
as $x=z_1z_2\cdots z_n=z'_1z'_2\cdots z'_{n'}$, we
may remove $\beta^{-1}(z'_1z'_2\cdots z'_{n'})$ from $Y$
without changing $X$. A finite number of these removals gives a trim decomposition. The set $Z'$ is obtained by removing all words in $Z$ which correspond to the letters no longer occurring in words in $Y'$ (we also remove such letters from $B$). The decomposition in Example \ref{exampleComposition} is not trim, since $aba = \beta(uw) = \beta(vu)$, but it can be made trim by taking $X = Y' \circ Z'$ with $Y' = \{u, uw\}$ and $Z' = \{a, ba\}$. In this case, $Y' \subset \{u, w\}^+$.

A set $X\subset A^*$ is \emph{complete} if any word in $A^*$ is a factor
of a word in $X^*$.

\begin{proposition}\label{propositionRedComp}
Let $Y\subset B^+$ and $Z\subset A^+$ be two composable
finite sets and let $X=Y\circ_\beta Z$ be a  trim decomposition. Let 
$\B=(Q,1,1)$
be the prefix
automaton of $Y$ and let $\Tr=(P,1,1)$ be the prefix
transducer of $Z$. The automaton $\A=\B\circ\Tr$
recognizes $X^*$ with multiplicities.

If $Y$ is complete,
there is a reduction $\rho$ from $\A$ onto the prefix
automaton of $Z$. Moreover, the automaton $\B$ can be identified
through $\beta$ with the restriction of $\A$ to $\rho^{-1}(1)$.
\end{proposition}
\begin{proof}
The simple paths in $\A$ have the form $(1,1)\edge{z_1}(b_1,1)\edge{z_2}(b_1b_2,1)\cdots\edge{z_n}(1,1)$ 
for $x=z_1\cdots z_n=\beta(b_1\cdots b_n)$ in $X$ and $z_i\in Z$.
Since the decomposition is trim, there is exactly one such path for every $x\in X$ and thus $\A$ recognizes $X^*$ with multiplicities.

Let us show that, if $Y$ is complete,
the map $\rho:(q,p)\to p$ is  a reduction
from $\A$ onto the prefix automaton of $Z$. 
We have to show that one has $p\edge{w}p'$ in the prefix automaton $\mathcal C$ of $Z$
if and only if there exist $q,q'\in Q$ such that $(q,p)\edge{w}(q',p')$.
Assume that $p\edge{w}p'$ in  $\mathcal C$. 
Then we have $p\edge{w|u}p'$ in the prefix transducer $\Tr$
for some $u\in B^*$. Since $Y$ is complete, there are some
$q,q'\in Q$ such that $q\edge{u}q'$ in $\B$. Then
$(q,p)\edge{w}(q',p')$ in $\A$. The converse is obvious.

Finally, the edges of the restriction of $\A$ to $\rho^{-1}(1)$
are the simple paths $(q,1)\edge{z}(q',1)$ for $z=\beta(b)\in Z$ 
and $q\edge{b}q'$ an edge of $\B$. This proves the last statement.
\end{proof}

\begin{figure}[hbt]
\centering
\includegraphics{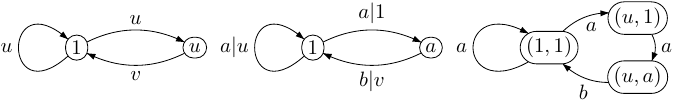}
\caption{The prefix automaton of $Y$, the prefix transducer $\Tr$ of $Z$
and the trim part of $\A$.}\label{figureFlower2}
\end{figure}

\begin{example}\label{example-composition}
Let $Y=\{u,uv\}$ and $Z=\{a,ab\}$ with $\beta:u\to a, v\to ab$.
We have, in view of Figure~\ref{figureFlower2},
\begin{displaymath}
\varphi_\A(a)=\begin{bmatrix}\varphi_\B(u)&I\\0&0\end{bmatrix}
\text{ and }
\varphi_\A(b)=\begin{bmatrix}0&0\\\varphi_\B(v)&0\end{bmatrix}.
\end{displaymath}

\end{example}

\section{Monoids of relations} \label{sect-relations}
We consider monoids of binary relations and prove some
results on idempotents and groups in such monoids.
Few authors have considered monoids of binary relations.
In~\cite{PlemmonsWest1970}, the Green's relations in the
monoid $\B_Q$ of all binary relations on a set $Q$ are considered.
It is shown in~\cite{MontaguePlemmons1969} that
any finite  group appears as a maximal subgroup
of $\B_Q$ (in contrast with the monoid of all partial
maps in which all maximal subgroups are symmetric groups).

We write indifferently relations on a set $Q$ as subsets
of $Q\times Q$, as boolean $Q\times Q$-matrices
or as directed graphs on a set $Q$ of vertices.

\enlargethispage{\baselineskip}
The \emph{rank} of a relation $m$ on $Q$ is the minimal cardinality
of a set $R$ such that $m=uv$ with $u$ a $Q\times R$ relation
and $v$ an $R\times Q$ relation. Equivalently, the
rank of $m$ is the minimal number of row (resp. column) vectors 
(which are possibly not rows or columns of $m$) which
generate over $\{0, 1\}$ the set of rows (resp. columns) of $m$.

For example, the full relation $m=Q\times Q$ has rank $1$. In terms
of matrices

\begin{displaymath}
m=\begin{bmatrix} 1\\1\\\vdots\\1\end{bmatrix}\begin{bmatrix}1&1&\cdots&1\end{bmatrix}
\end{displaymath}

More generally, the rank of an equivalence relation is equal to the number of its classes.

A \emph{fixed point} of a relation $m$ on $Q$ is
an element $q\in Q$ such that $q\longedge{m}q$.
The following result appears in~\cite{Schutzenberger1979}
(see also~\cite{LeRest1980}).

\begin{proposition}\label{propositionIdempotent}
Let $e$ be an idempotent relation on a finite set $Q$,
let $S$ be the set of fixed points of $e$ and let
$\Gamma$ be the set of strongly connected components 
of the restriction of $e$ to $S$. 
\begin{enumerate}
\item For all $p,q\in Q$ we have $p\longedge{e} q$ if and only
if there exists an $s\in S$ such that $p\longedge{e}s$ and $s\longedge{e}q$.
\item We have
\begin{equation}
e=\ell r  \label{eqColumnRow}
\end{equation}
where $\ell=\{(p,\sigma)\in Q\times\Gamma\mid p\longedge{e}s\text{ for some }
s\in \sigma\}$ and $r=\{(\sigma,q)\in \Gamma\times Q\mid s\longedge{e}q\text{ for some }
s\in \sigma\}$.
\end{enumerate}
\end{proposition}
\begin{proof}
1. Choose $n>\Card(Q)$. Since $p\longedge{e^n}q$, there
is some $s\in Q$ such that $p\longedge{e^i}s\longedge{e^j}s\longedge{e^k}q$
with $i+j+k=n$. Then $p\longedge{e}s\longedge{e}s\longedge{e}q$
and the statement is proved. The other direction is obvious.

2. If $p\longedge{e}q$, let
$s\in S$ be such that $p\longedge{e}s\longedge{e}q$
and let $\sigma$ be the strongly connected component of $s$. Then
$p\longedge{\ell}\sigma\longedge{r}q$. Thus $e\le\ell r$, which means that each element of $e$ is not larger than the corresponding element of $\ell r$ when these relations are considered as binary matrices. Conversely,
if $p\longedge{\ell}\sigma\longedge{r}q$ there are $s,s'\in \sigma$
such that $p\longedge{e}s$ and $s'\longedge{s'}q$. Since
$s,s'$ are in the same stongly connected component, we have $s\edge{e}s'$
and we obtain $p\edge{e}s\edge{e}s'\edge{e}q$,
whence $p\edge{e}q$.
\end{proof}

The decomposition of $e=lr$ given by Equation~\eqref{eqColumnRow}
is called the \emph{column-row decomposition} of $e$.
Note that Proposition~\ref{propositionIdempotent}
is false without the finiteness hypothesis on $Q$.
Indeed, the relation $e=\{(x,y)\in \R^2\mid x<y\}$ is idempotent,
but has no fixed points.

\begin{example}
The matrix
\begin{displaymath}
m=\begin{bmatrix}1&1&1&0\\1&1&1&0\\0&0&0&0\\1&1&1&0\end{bmatrix}
\end{displaymath}
is an idempotent of rank $1$.
\end{example}

For an element $m$ of a monoid $M$, we denote by $H(m)$
the $\GH$-class of $m$, where $\GH$ is the Green relation
$\GH=\GR\cap\GL$ (see~\cite{BerstelPerrinReutenauer2009} for the definitions). It is a group if and only if it
contains an idempotent $e$ (see~\cite{BerstelPerrinReutenauer2009}). In this case, every
$m\in H(e)$ has a unique inverse $m^{-1}$ in the
group $H(e)$. 

The following result is the transposition of Proposition 9.1.7 in \cite{BerstelPerrinReutenauer2009} to arbitrary
monoids of relations.
However, the result is restricted to a statement on the
group $H(e)$ instead of the monoid~$eMe$.

\begin{proposition}\label{prop-idempotents}
Let $M$ be a monoid of relations on a finite set $Q$,
let $e\in M$ be idempotent and let $\Gamma$ be the set of strongly connected
components of the fixed points of $e$. For $m\in H(e)$, let $\gamma_e(m)$ be the relation on $\Gamma$ defined by 
\begin{displaymath}
\gamma_e(m)=\{(\rho,\sigma)\in \Gamma\times \Gamma\mid r\edge{m}s\edge{m^{-1}}r\text{ for some
$r\in \rho$ and $s\in\sigma$}\}
\end{displaymath}
Then $m\mapsto \gamma_e(m)$ is an isomorphism from $H(e)$ onto a group of permutations on $\Gamma$.
\end{proposition}
\begin{proof}
First, $\gamma_e(m)$ is a map. Indeed, let $s\edge{m}t\edge{m^{-1}}s$
and $s'\edge{m}t'\edge{m^{-1}}s'$. If $s\edge{e}s'$, we have
$t\edge{m^{-1}}s\edge{e}s'\edge{m}t'$ and thus $t\edge{e}t'$.
By a symmetrical proof, we obtain that $\gamma_e(m)$ is a permutation.

Next, it is easy to verify that $\gamma_e$ is a morphism.

Finally, $\gamma_e$ is injective. Indeed, assume that for $m,m'\in H(e)$
we have $\gamma_e(m)=\gamma_e(m')$. Suppose that $p\edge{m}q$. 

Assume
first that  $p$
is a fixed point of $e$. Let $r,r'$ be such that $p\edge{m}r\edge{m^{-1}}p$
and $p\edge{m'}r'\edge{m'^{-1}}p$. Since $\gamma_e(m)=\gamma_e(m')$,
we obtain that $r,r'$ are in the same element of $\Gamma$. We
conclude that $p\edge{m'}r'\edge{e}r\edge{m^{-1}}p\edge{m}q$
which implies that $p\edge{m'}q$.

Now if $p$ is not a fixed point of $e$, since $em=m$, there
is an $r$ such that $p\edge{e}r\edge{m}q$.
By Proposition~\ref{propositionIdempotent}, there is a fixed point $r'$ of $e$ such that $p\edge{e}r'\edge{e}r\edge{m}q$. Then
$r'\edge{m}q$ implies $r'\edge{m'}q$ by the preceding argument,
and finally $p\edge{m'}q$.	
\end{proof}

We denote  $G_e=\gamma_e(H(e))$. The definition of $\gamma_e$ can be formulated differently.

\begin{proposition}\label{defEquiv}
Let $M$ be a monoid of relations on a finite set $Q$
and let $e\in M$ be an idempotent. Let $\sigma,\tau$
be two distinct connected components of fixed points of $e$ and let $s\in\sigma,t\in\tau$. If $e_{s,t}=1$, then
$m_{t,s}=0$ for every $m\in H(e)$ and thus $(\sigma,\tau)\notin\gamma_e(m)$.
If $e_{s,t}=e_{t,s}=0$ then $s\edge{m}t$ implies $(\sigma,\tau)\in \gamma_e(m)$.
\end{proposition}
\begin{proof}
Assume first that $e_{s,t}=1$ so that the restriction
of $e$ to $\{s,t\}$ is the matrix $\begin{bmatrix}1&1\\0&1\end{bmatrix}$. If $m_{t,s}=1$, then the restriction of $m$
to $\{s,t\}$ is the matrix with all ones, which is impossible 
since no power of $m$ can be equal to $e$.
If the restriction of $e$ to $\{s,t\}$ is the identity, then
the restriction of $m\in H(e)$ is a permutation. Thus
$(\sigma,\tau)\in\gamma_e(m)$ if and only if $s\edge{m}t$.
\end{proof}

The following extends Proposition 9.1.9 in \cite{BerstelPerrinReutenauer2009}.
It uses the Green relation $\GD=\GL\GR=\GR\GL$. Two permutation groups $G$ over $Q$ and $G'$ over $Q'$ are called \emph{equivalent} if there exists a bijection $\alpha:Q \to Q'$ and an isomorphism $\psi: G \to G'$ such that for all $q \in Q$ and $g \in G$ we have $\alpha(q.g) = \alpha(q).\psi(g)$, where $q.g$ is the action of $g$ on $q$ (see Section 1.13 of \cite{BerstelPerrinReutenauer2009}). In a more standard terminology, two permutation groups are equivalent if and only if their group actions are isomorphic, though we use the terminology of \cite{BerstelPerrinReutenauer2009} to simplify the comparison with the results described there.

\begin{proposition}\label{prop-groups-equiv}
Let $M$ be a monoid of relations on a finite set $Q$
and let $e,e'\in M$ be $\GD$-equivalent idempotents. Then
the groups $G_e$ and $G_{e'}$ are equivalent permutation groups.
\end{proposition}
\begin{proof}
Let $(a,a',b,b')$ be a passing system from $e$ to $e'$, that is such that
\begin{displaymath}
eaa'=e,\quad bb'e'=e',\quad ea=b'e'.
\end{displaymath}
We will verify that there is a commutative diagram of isomorphisms shown in Figure~\ref{diagram}.

\begin{figure}[hbt]
\centering
\includegraphics{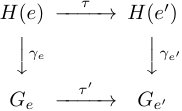}
\caption{Commutative diagram of isomorphisms.}\label{diagram}
\end{figure}

We define the map $\tau$ by $\tau(m)=bma$. Then it is easy to verify
that $\tau$ is a morphism and that $m'\mapsto b'm'a'$ is its inverse.
Thus $\tau$ is an isomorphism.

We define $\tau'$ as follows. Let 
$\Gamma_e,\Gamma_{e'}$ be the
sets of strongly connected components of fixed points of $e$
and $e'$ respectively. Let $\theta$ be the relation between $\Gamma_e$
and $\Gamma_{e'}$ defined by $(\sigma,\sigma')\in\theta$ if
for some $s\in \sigma$ and $s'\in\sigma'$, we have $s\edge{eae'}s'$.
One may verify that $\theta$ is a bijection between $\Gamma_e$ and $\Gamma_{e'}$.
Its inverse is the map on classes induced by $e'be=e'a'e$.
Then $\tau'(n)=\theta^tn\theta$.

We verify that the diagram is commutative. 
Suppose  that for some $m\in H(e)$
$(\sigma'_1,\sigma'_1)\in\tau'(\gamma_e(m))$.
By definition of $\tau'$ there exist $\sigma_1,\sigma_2\in\Gamma_e$
such that
\begin{displaymath}
(\sigma'_1,\sigma_1)\in\theta^t,\quad 
(\sigma_1,\sigma_2)\in \gamma_e(m)
\text{ and }(\sigma_2,\sigma'_2)\in\theta.
\end{displaymath}
Then for $s_1\in \sigma_1$, $s'_1\in\sigma'_1$, $s'_2\in\sigma'_2$
and $s_2\in \sigma_2$, we have
\begin{displaymath}
s'_1\edge{e'be}s_1,\quad
s_1\edge{m}s_2\edge{m^{-1}}s_1. \quad
s_2\edge{eae'}s'_2.
\end{displaymath} Then
$s'_1\edge{bma}s'_2\edge{bm^{-1}a}s'_1$ showing that 
$(\sigma'_1,\sigma'_1)\in\gamma_{e'}(\tau(m))$.
\end{proof}

Note that, contrary to the case of a monoid of unambiguous relations,
two $\GD$-equivalent idempotents need not have the
same number of fixed points, as shown by the following example.

\begin{example}
Let $M$ be the monoid of all relations on $Q=\{1,2\}$.
The two idempotents
\begin{displaymath}
e=\begin{bmatrix}1&0\\0&0\end{bmatrix}
,\quad e'=\begin{bmatrix}1&1\\1&1\end{bmatrix}
\end{displaymath}
are $\GD$-equivalent although the first has one fixed point and the second has two.
\end{example}

Let $M$ be a monoid of relations on a finite set $Q$. The \emph{minimal
rank} of $M$, denoted $r(M)$ is the minimum of the ranks of the
elements of $M$ other than $0$. The following
statement generalizes Theorem 9.3.10 in \cite{BerstelPerrinReutenauer2009} from unambiguous to arbitrary transitive monoids of relations. A $\GD$-class is \emph{regular} if it contains an idempotent. A monoid of relations on $Q$
is \emph{transitive} if for every $p,q\in Q$, there
is an $m\in M$ such that $p\edge{m}q$.

\begin{theorem}\label{theorem9.3.10}
Let $M$ be a transitive monoid of relations on a finite set $Q$.
The set $K$
of elements of rank $r(M)$ is a regular $\GD$-class. The groups $G_e$ for $e$ idempotent in $K$
are equivalent transitive permutation groups. Moreover, for a fixed point $i$ of $e$, the minimal rank $r(M)$ is the index of the subgroup 
$\{m\in H(e)\mid i\edge{m}i\}$ in $H(e)$.
\end{theorem}

\enlargethispage{1.5\baselineskip}
\begin{proof}
The proof is the same as for the case of an unambiguous monoid 
of relations except for the last statement.
Let $\sigma,\tau$ be two distinct strongly connected components of
fixed points of $e$ and let $s\in \sigma, t\in\tau$. Since $M$ is transitive
there is an $m\in M$ such that $s\edge{m}t$. Then $eme$ is not $0$
and thus $eme\in H(e)$. Similarly, if $n\in M$ is such that
$t\edge{n}s$, then $ene\in H(e)$. This implies by Proposition~\ref{defEquiv}
that
the restriction of $e$ to $\{s,t\}$ is the identity and that
$(\sigma,\tau)\in \gamma_e(eme)$. Thus $G_e$ is transitive.
The last statement follows from the fact that for any transitive
permutation group on a set $S$, the number of elements of $S$ is equal
to the index of the subgroup fixing one of the points of $S$ (Proposition 1.13.2 of \cite{BerstelPerrinReutenauer2009}).
\end{proof}

The \emph{Suschkevitch group} of $M$ is one of the equivalent
groups $G_e$ for $e$ of rank $r(M)$.

\section{Group and degree of a set} \label{sect-group}
Let $\A=(P,i,i)$ and $\B=(Q,j,j)$
be  automata and let $\rho:P\to Q$ be a reduction. For $m=\varphi_\A(w)$,
the relation $n=\varphi_\B(w)$ is well defined. We denote it by
$n=\hat{\rho}(m)$. Then $\hat{\rho}$ is a morphism from $\varphi_\A(A^*)$
onto $\varphi_\B(A^*)$ called the \emph{morphism associated} with $\rho$.
The following result extends Proposition 9.5.1 in \cite{BerstelPerrinReutenauer2009} to arbitrary finite sets of words.

\begin{proposition}\label{proposition9.5.1}
Let $X\subset A^+$ be finite. Let $\A=(P,i,i)$ and $\B=(Q,j,j)$
be trim automata recognizing $X^*$. Let $M=\varphi_\A(A^*)$
and $N=\varphi_\B(A^*)$. Let $E$ be the set of idempotents in $M$
and $F$ the set of idempotents in $N$.

Let $\rho$ be a sharp reduction of $\A$ onto $\B$ and let
$\hat{\rho}:M\to N$ be the morphism associated with $\rho$. Then
\begin{enumerate}
\item $\hat{\rho}(E)=F$.
\item Let $e\in E$ and $f=\hat{\rho}(e)$. The restriction of $\rho$
to the set $S$ of fixed points of $e$ is a bijection on the
set of fixed points of $f$, and the groups $H_e$ and $H_f$ are equivalent.
\end{enumerate}
\end{proposition}
\begin{proof}
1. Let $e\in E$. Then $\hat{\rho}(e)$ is idempotent since $\hat{\rho}$
is a morphism. Thus $\hat{\rho}(E)\subset F$. Conversely, if $f\in F$, let $w\in A^*$ be such that
$\varphi_\B(w)=f$. Let $n\ge 1$ be such that $e=\varphi_\A(w)^n$ is
idempotent. Then $\hat{\rho}(e)=f$ since $\hat{\rho}\circ\varphi_\A=\varphi_\B$.

2. Let $S$ be the set of fixed points of $e$ and $T$ the set of fixed points of $f$. Consider $s\in S$ and let $t=\rho(s)$. From $s\edge{e}s$,
we obtain $t\edge{f}t$ and thus $\rho(S)\subset T$. Conversely, let $t\in T$.
The restriction of $e$ to the set $R=\rho^{-1}(t)$ is a non zero idempotent.
Thus there is some $s\in R$ which is a fixed point of this idempotent,
ans thus of $e$. Thus $t\in \rho(S)$.

Since $\hat{\rho}$ is a morphism from $M$ onto $N$,
we have $\hat{\rho}(H(e))=H(f)$.
It is clear that $\rho$ maps a strongly connected component
of $e$ on a strongly connected component of $f$. To show that
this map is a bijection, consider $s,s'\in S$ such that
$\rho(s),\rho(s')$ belong to the same connected component.
We may assume that $e$ is not the equality relation.
Let $w\in A^+$ be such that $\varphi_\A(w)=e$. Since
$X$ is finite, there are factorizations $w=uv=u'v'$
such that $s\edge{u}i\edge{v}s$ and $s'\edge{u'}i\edge{v'}s'$.
Then we have $j\edge{v}\rho(s)\edge{w}\rho(s')\edge{u'}j$.
Since $\rho$ is sharp, this implies $i\edge{vwu'}i$
and finally $s\edge{uvwu'v'}s'$. This shows
that $s\edge{e}s'$. A similar proof shows that
$s'\edge{e}s$. Thus,  $s,s'$
belong to the same connected component of $e$.

Moreover,
for every $m\in H(e)$, one has $s\edge{m}t\edge{m^{-1}}s$
if and only if 
$\rho(s)\edge{\hat{\rho}(m)}\rho(t)\edge{\hat{\rho}(m^{-1})} \rho(s)$.
Thus $H_e$ and $H_f$ are equivalent permutation groups.
\end{proof}

Let $X\subset A^+$ be a finite set and let $\A$ be the flower automaton of $X$.
The \emph{degree} of $X$, denoted $d(X)$ is the minimal rank of the monoid
$M=\varphi_\A(A^*)$.
The \emph{group of} $X$ is the Suschkevitch group of $M$.
Proposition~\ref{proposition9.5.1} shows that the definitions
of the group and of the degree do not depend on the automaton
chosen to recognize $X^*$, provided one takes a trim
automaton recognizing $X^*$ with multiplicities.

\section{Synchronization} \label{sect-synchronization}

Let $X\subset A^+$ be a finite set of words.
A  word $x\in A^*$ is \emph{synchronizing} for $X$ if
for every $u,v\in A^*$, $uxv\in X^*\Rightarrow ux,xv\in X^*$.
A set $X$ is \emph{synchronizing} if there is a synchronizing word $x\in X^*$. The next proposition generalizes Proposition 10.1.11 of \cite{BerstelPerrinReutenauer2009}

\begin{proposition}\label{prop-degree-one}
A finite set $X \subset A^+$ is synchronizing if and only if
its degree $d(X)$ is~$1$.
\end{proposition}
\begin{proof}
Let $\A=(Q,i,i)$ be a trim automaton recognizing $X^*$.

Assume first that $d(X)=1$. Let $x\in X^*$ be such that
$\varphi_\A(x)$ has rank $1$. If $uxv\in X^*$, we have
$i\edge{u}p\edge{x}q\edge{v}i$ for some $p,q\in Q$.
Since $\varphi_\A(x)$ has rank $1$, we deduce from
$i\edge{x}i$ and $p\edge{x}q$ that we have also
$i\edge{x}q$ and $p\edge{x}i$. Thus $ux,xv\in X^*$, showing
that $x$ is synchronizing.

Assume conversely that $X$ is synchronizing. Let $x\in X^*$ be a synchronizing
word. Replacing $x$ by some its power, we may assume that $\varphi_\A(x)$
is an idempotent $e$.
Let $m\in H(e)$ and let $w\in\varphi_\A^{-1}(m)$. Since $H(e)$ is finite, there is some $n\ge 1$
such that $m^n=e$. Then $(me)^{n}=e$ implies
that $(wx)^{n}\in X^*$. Since $x$ is synchronizing, we obtain
$wx\in X^*$ and since $\varphi_\A(wx)=me=m$, this implies $w\in X^*$.
This shows that $H(e)$ is contained in $\varphi_\A(X^*)$
and thus that $d(X)=1$ by Theorem~\ref{theorem9.3.10}.	
\end{proof}

\begin{example}\label{exampleaba}
Consider again $X=\{a,ab,ba\}$ (Example~\ref{example{a,ab,ba}}).
The flower automaton of $X$ is represented again for convenience
in Figure~\ref{figureFlower3} (left).

The minimal rank of the elements of $\varphi_\A(A^*)$
is $1$. Indeed, we have
\begin{displaymath}
\varphi_\A(a^2)=\begin{bmatrix}1&1&0\\0&0&0\\1&1&0\end{bmatrix}
=\begin{bmatrix}1\\0\\1\end{bmatrix}\begin{bmatrix}1&1&0\end{bmatrix}
\end{displaymath}
Accordingly, $aa$ is a synchronizing word.

\begin{figure}[hbt]
\centering
\includegraphics{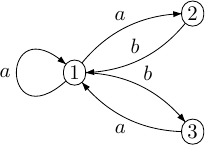}
\hspace{2cm}
\includegraphics{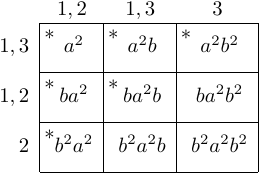}
\caption{The flower automaton of $X$ (left) and the set $K$ of elements of rank $1$ (right).}\label{figureFlower3}
\end{figure}

The set $K$ of elements of rank $1$ is represented in Figure~\ref{figureFlower3} (right).
For each $\GH$-class, we indicate on its left the set of states $p$
such that the row of index $p$ in nonzero. Similarly, we indicate
above it the set of states $q$ such that the column of index $q$ is nonzero.
A star $*$ indicates an $\GH$-class which is a group. Note that
\begin{displaymath}
\varphi_\A(a^2b)=\begin{bmatrix}1&0&1\\0&0&0\\1&0&1\end{bmatrix}
\end{displaymath}
has two fixed points but only one strongly connected class, in agreement
with fact that it is of rank $1$.
\end{example}

\section{Groups and composition} \label{sect-groups-composition}
Given a transitive 
permutation group $G$ on a set $Q$, an \emph{imprimitivity relation} of $G$
is an equivalence on $Q$ compatible with the group action.
If $\theta$ is such an equivalence relation, we denote by
$G_\theta$ the permutation group induced by the action of $G$
on the classes of $\theta$. The groups induced by the action
on the class of an element $i\in Q$ by the action of the elements of $G$ stabilizing the class of $i$ are all equivalent. We denote by $G^\theta$ one of them.
For two permutation groups $G,H$ on sets $P$ and $Q$ respectively,
we denote $G\le H$ if there is an imprimitivity equivalence $\theta$
on $Q$ such that $G=H_\theta$.

The next theorem generalizes Proposition 11.1.2 of \cite{BerstelPerrinReutenauer2009}. 

\begin{theorem}\label{theoremDegrees}
Let $X\subset A^+$ be a finite set with a trim decomposition $X=Y\circ Z$,
where $Y$ is complete. There exists an imprimitivity equivalence
$\theta$ of $G=G(X)$ such that
\begin{displaymath}
G^\theta\le G(Y),\quad G _\theta=G(Z).
\end{displaymath}
In particular, $d(X)\le d(Y) \cdot d(Z)$.
\end{theorem}
\begin{proof}
Let  $\B=(Q,i,i)$ be the flower automaton of $Y$
and let $\Tr$ be the prefix transducer of $Z$.
Let $\A=\B\circ \Tr$.
By Proposition~\ref{propositionRedComp}, there is a reduction
$\rho$ from $\A=(Q\times P,(i,1),(i,1))$ onto the prefix automaton $\mathcal C$ of $Z$.

Let $e$ be an idempotent of minimal rank  in $\varphi_\A(X^*)$.
Let $S$ be the set of fixed points of $e$ and
let $\Gamma$ be the set of connected components (scc) of the elements of $S$. 
Let $\hat{S}$ be the set of fixed points of $\hat{e}=\hat{\rho}(e)$
and let $\hat{\Gamma}$ be the set of corresponding scc's. If $s,s'\in S$ are in the same scc, then $\rho(s),\rho(s')$ are in the same scc
of $\hat{S}$. Thus, we have a well-defined
map $\bar{\rho}:\Gamma\to\hat{\Gamma}$ such that $s\in\Gamma$
if and only if $\rho(s)\in \bar{\rho}(\Gamma)$.

We define an equivalence
$\theta$ on $\Gamma$ by $\sigma\equiv \sigma'$ if 
$\bar{\rho}(\sigma)=\bar{\rho}(\sigma')$. Let $m\in H(e)$ and suppose that
$(\sigma,\tau),(\sigma',\tau')\in \gamma_e(m)$.
If $\sigma\equiv\sigma'\bmod\theta$,
then $\tau\equiv\tau'\bmod\theta$. Let indeed
$s\in \sigma,s'\in\sigma'$ and $t\in\tau,t'\in\tau'$.
We have by definition of $\gamma_e$
\begin{displaymath}
s\edge{m}t\edge{m^{-1}}s
\text{ and }
s'\edge{m}t'\edge{m^{-1}}s'
\end{displaymath}
and thus
\begin{displaymath}
\rho(s)\edge{\hat{\rho}(m)}\rho(t)\edge{\hat{\rho}(m)^{-1}}\rho(s)
\text{ and }
\rho(s')\edge{\hat{\rho}(m)}\rho(t')\edge{\hat{\rho}(m)^{-1}}\rho(s')
\end{displaymath}
This implies that $\rho(t)\edge{\hat{e}}\rho(t')$
and  $\rho(t')\edge{\hat{e}}\rho(t)$. But since $\gamma_{\hat{e}}(\hat{m})$
is a permutation, this forces $\bar{\rho}(\tau)=\bar{\rho}(\tau')$ and finally
$\tau\equiv\tau'\bmod\theta$. 
Since the action of $H(e)$ on the classes of $\theta$ is
the same as the action of $H(\hat{e})$, we have $G(Z)=G_\theta$.

Finally, let $\sigma\in\Gamma$ be the class of the initial state $(i,1)$
and let $I$ be its class $\bmod\theta$. Thus $d(X)=\Card(I)d(Z)$.
Let $x\in X^*$ be such that
$\varphi_\A(x)=e$ and let $y=\beta^{-1}(x)$. Then $f=\varphi_\B(y)$
is an idempotent of $\varphi_\B(B^*)$ of rank $d(Y)$.
Let $U$ be the set of fixed points of $f$ and let $\Phi$
be the set of scc of $U$ for the action of $f$.
Let $\sigma$ be the equivalence on $\Phi$ induced by the equivalence
$r\equiv s$ if $(r,1),(s,1)$ belong to the same scc for $e$.
Then $\sigma$ is an imprimitivity equivalence for $G(Y)$
such that $G(Y)_\sigma=G^\theta$. Thus
$G^\theta\le G(Y)$ and $\Card(I)\le d(Y)$,
which implies $d(X)\le d(Y) \cdot d(Z)$.
\end{proof}

\enlargethispage{-1\baselineskip}
\begin{example}
  Let $Z=\{a,ab,ba,ca\}$ and $X=Z^2$. We have $X=Y\circ_\beta Z$
  with $Y=\{u,v,w,x\}^2$ and
  $\beta:u\mapsto a, v\mapsto ab, w\mapsto bc,x\mapsto ca$.
  The word $aa$ is synchronizing for $Z$ and thus $d(Z)=1$.
  In contrast, we have $d(Y)$ and $G(Y)=\Z/2\Z$.
  It can be verified that the word $ca^2b$ is synchronizing
  for $X$ and thus $d(X)=1$. Thus $d(X)<d(Y)\cdot d(Z)=2\cdot 1=2$.
  Thus the case of a stict inequality can occur. This
  is made possible by the fact that $Z$ is not a code.
  Indeed, we have $(ab)(ca)=a(bc)a$.
\end{example}

\section{Decompositions of codes} \label{sect-codes}

Finally, we use the developed techniques to show that for a uniquely decipherable code $X$ for all the trim decompositions of the form $X = Y \circ Z$ with $Y$ is complete we have that $Z$ (and thus $Y$) is a uniquely decipherable code as well. It shows that, as long as we require $Y$ to be complete, we do not get any new trim decompositions of uniquely decipherable codes even if we decompose them as arbitrary sets of words.

\begin{proposition}\label{propositionCompleteDecomp}
Let $X=Y\circ Z$ be a trim decomposition of a finite set $X$. If $X$
is a uniquely decipherable code and if $Y$
is complete, then $Z$ is a uniquely decipherable code.
\end{proposition}
\begin{proof}
Since $\beta$ is trim, $Y$ is a uniquely decipherable code.
Let $\beta:B\rightarrow Z$ be the coding morphism for $Z$
such that $X=Y\circ_\beta Z$.
Assume that $z\in Z^*$ is a word with more than one factorization
into words of $Z$. Let $u,v\in B^*$ two distinct
elements   in $\beta^{-1}(z)$. Let $\A$ be the flower automaton of $Y$.
Let $y\in Y^*$ be such that $\varphi_\A(y)$ has minimal rank. Then
$yuy,yvy$ are not zero since $Y$ is complete. Thus $\varphi_\A(yuy),\varphi_\A(yvy)$
belong to the $\GH$-class of $\varphi_\A(y)$ which is a finite
group. Let $e$ be its idempotent. There
are integers $n,m,p$ such that $\varphi_\A(y)^n=\varphi_\A(yuy)^m
=\varphi_\A(yvy)^p=e$.
Since $y\in Y^*$, this implies that $e\in\varphi_\A(Y^*)$
and thus that $(yuy)^m,(yvy)^p$ are in $Y^*$.
We conclude that $Y$ is not a uniquely decipherable code, a contradiction.
\end{proof}

This is false if we do not require $Y$ to be complete. Consider a code $X = \{ab, abaab, abbab\}$, which can be decomposed into $X = Y \circ Z$ with $Y = \{u, uvu, uwu\}$ and $Z = \{ab, a, b\}$. The decomposition is obviously trim, the set $X$ is a uniquely decipherable code, but the set $Z$ is not a uniquely decipherable code.

\bibliography{finGenSubRev}

\end{document}